\documentclass[11pt,letterpaper]{article}

\usepackage{typearea}
\typearea{14}

\usepackage{comment,algorithm,algorithmic,multicol}


\usepackage{color,colortbl}
\usepackage{float}
\usepackage{amsthm}
\usepackage{amsmath}
\usepackage{amssymb}
\usepackage{graphicx}
\usepackage{xspace}
\usepackage{nicefrac}
\usepackage[colorinlistoftodos,prependcaption,textsize=tiny]{todonotes}

\usepackage{thmtools}
\usepackage{thm-restate}
\declaretheorem[name=Lemma]{lemma}

\definecolor{Darkblue}{rgb}{0,0,0.4}
\definecolor{Brown}{cmyk}{0,0.61,1.,0.60}
\definecolor{Purple}{cmyk}{0.45,0.86,0,0}

\ifx\pdftexversion\undefined
\usepackage[colorlinks,linkcolor=black,filecolor=black,citecolor=black,urlco
lor=black,pdfstartview=FitH]{hyperref}
\else
\usepackage[colorlinks,linkcolor=Darkblue,filecolor=blue,citecolor=blue,urlcolor=Darkblue,pdfstartview=FitH]{hyperref}
\fi

\newtheorem{theorem}{Theorem}
\newtheorem{corollary}{Corollary}

\newtheorem{claim}{Claim}

\newtheorem{fact}{Fact}

\newcommand{\namedref}[2]{\hyperref[#2]{#1~\ref*{#2}}}
\newcommand{\sectionref}[1]{\namedref{Section}{#1}}
\newcommand{\appendixref}[1]{\namedref{Appendix}{#1}}

\newcommand{\theoremref}[1]{\namedref{Theorem}{#1}}

\newcommand{\figureref}[1]{\namedref{Figure}{#1}}
\newcommand{\claimref}[1]{\namedref{Claim}{#1}}
\newcommand{\lemmaref}[1]{\namedref{Lemma}{#1}}

\newcommand{\corollaryref}[1]{\namedref{Corollary}{#1}}

\newcommand{\algref}[1]{\namedref{Algorithm}{#1}}

\newcommand{\R}{\mathbb{R}}

\newcommand{\Exp}{\mathsf{Exp}}

\newcommand{\eps}{\epsilon}

\newcommand{\la}{~\leftarrow~}

\newcommand{\cCE}{c_{\text{\tiny CE}}}
\newcommand{\cint}{c_{\text{\tiny int}}}
\newcommand{\ECE}{\mathcal{E}^{\text{\tiny CE}}}
\newcommand{\XCE}{\mathcal{X}^{\text{\tiny CE}}}
\newcommand{\ECUB}{\mathcal{E}^{\text{\tiny CUB}}}
\newcommand{\EfBig}{\mathcal{E}^{\text{\tiny fBig}}}

\def\inline#1:{\par\vskip 7pt\noindent{\bf #1:}\hskip 10pt}

\def\inline#1:{\par\vskip 7pt\noindent{\bf #1:}\hskip 10pt}

\def\blackslug{\hbox{\hskip 1pt \vrule width 4pt height 8pt
		depth 1.5pt \hskip 1pt}}

\def\QED{\quad\blackslug\lower 8.5pt\null\par}

\newcommand{\initOneLiners}{%
	\setlength{\itemsep}{0pt}
	\setlength{\parsep }{0pt}
	\setlength{\topsep }{0pt}
}

\newcommand{\alert}[1]{\textbf{\color{red}
		[[[#1]]]}\marginpar{\textbf{\color{red}**}}\typeout{ALERT:
		\the\inputlineno: #1}}
\definecolor{purple}{rgb}{0.294, 0, 0.71}

\floatstyle{ruled}
\newfloat{algorithm}{tbp}{loa}
\providecommand{\algorithmname}{Algorithm}
\floatname{algorithm}{\protect\algorithmname}

\usepackage{authblk}
\usepackage{pdfsync}

\title{Steiner Point Removal with Distortion $O(\log k)$}
\author{Arnold Filtser\\Ben Gurion University of the Negev\\
	Email: \texttt{arnoldf@cs.bgu.ac.il}}
\date{\today}

\begin{document}
	\maketitle
	\thispagestyle{empty}
	\nonumber
	\begin{abstract}
			In the Steiner point removal (SPR) problem, we are given a weighted graph $G=(V,E)$ and a set of terminals $K\subset V$ of size $k$.
			The objective is to find a minor $M$ of $G$ with only the terminals as its vertex set, such that the distance between the terminals will be preserved up to a small multiplicative distortion.
			Kamma, Krauthgamer and Nguyen \cite{KKN15} used a ball-growing algorithm with exponential distributions to show that
			the distortion is at most $O(\log^5 k)$.
			Cheung \cite{Cheung17} improved the analysis of the same algorithm, bounding the distortion by $O(\log^2 k)$.
			We improve the analysis of this ball-growing algorithm even further, bounding the distortion by $O(\log k)$.
	\end{abstract}

\newpage
\pagenumbering{arabic}
\section{Introduction}
In graph compression problems the input is usually a massive graph. The objective is to compress the graph into a smaller graph, while preserving certain  properties of the original graph such as distances or cut values.
Compression allows us to obtain faster algorithms, while reducing the storage space. In the era of massive data, the benefits are obvious.
Examples of such structures are graph spanners \cite{PS89}, distance oracles \cite{TZ05}, cut sparsifiers \cite{BK96}, spectral sparsifiers \cite{BSS12}, vertex sparsifiers \cite{Moitra09} and more.

In this paper we study the \emph{Steiner point removal} (SPR) problem. Here we are given an undirected graph $G=(V,E)$ with positive weight function $w:E\rightarrow\mathbb{R}_+$, and a subset of terminals $K\subseteq V$ of size $k$.
The goal is to construct a new graph $M=(K,E')$ with positive weight function $w'$, with the terminals as its vertex set, such that: (1) $M$ is a graph minor of $G$, and (2) the distance between every pair of terminals $t,t'$ is distorted by at most a multiplicative factor of $\alpha$, formally
$$\forall t,t'\in K,~~d_G(t,t')\le d_{M}(t,t')\le \alpha \cdot d_G(t,t')~.$$
Property (1) expresses preservation of the topological structure of the original graph. For example if $G$ was planar, so will $M$ be. Whereas property (2) expresses preservation of the geometric structure of the original graph, that is, distances between terminals.
The question is: what is the minimal $\alpha$ (which may depend on $k$) such that every graph with a terminal set of size $k$ will have a solution to the SPR problem with distortion $\alpha$.

The first one to study a problem of this flavor was Gupta \cite{G01}, who showed that given a weighted tree $T$ with a subset of terminals $K$, there is a tree $T'$ with $K$ as its vertex set, that preserves all the distances between terminals up to a multiplicative factor of $8$.
Chan, Xia, Konjevod, and Richa \cite{CXKR06}, observed that the tree $T'$ of Gupta is in fact a minor of the original tree $T$. They showed that $8$ is the best possible distortion, and formulated the problem for general graphs.
This lower bound of $8$ is achieved on the complete unweighted binary tree, and is the best known lower bound for the general SPR problem. 

Basu and Gupta \cite{BG08} showed that on outerplanar graphs, the SPR problem can be solved with distortion $O(1)$.

Kamma, Krauthgamer and Nguyen were the first to bound the distortion for general graphs.
They suggested a natural ball growing algorithm.   
Their first analysis provide  $O(\log^6 k)$ distortion (conference version \cite{KKN14}), which they later improved to $O(\log^5 k)$  (journal version \cite{KKN15}).
Very recently, Cheung \cite{Cheung17} improved the analysis of the same algorithm further, providing an $O(\log^{2}k)$ upper bound on the distortion.

The main contribution of this paper is an even further improvement upon the analysis of the same algorithm, providing an $O(\log k)$ upper bound for the SPR problem on general graphs.
Closing the gap between the lower bound of $8$ to the upper bound of $O(\log k)$ remains an intriguing open question.

\subsection{Related Work}
Englert et. al. \cite{EGKRTT14} showed that every graph $G$, admits a distribution $\mathcal{D}$ over terminal minors with expected distortion $O(\log k)$. Formally, for all $t_{i},t_{j}\in K$, it holds that $1\le\frac{\mathbb{E}_{M\sim\mathcal{D}}\left[d_{M}(t_{i},t_{j})\right]}{d_{G}(t_{i},t_{j})}\le O\left(\log k\right)$. 
Thus, \theoremref{thm:mainSPR} can be seen as improvement upon \cite{EGKRTT14}, where we replace distribution with a single minor.
Englert et. al. showed better results for  $\beta$-decomposable graphs, in particular showing that graphs excluding a fixed minor, admitting a distribution with $O(1)$ expected distortion.

Krauthgamer, Nguyen and Zondiner \cite{KNZ14} showed that if we allowing the minor $M$ to contain at most ${k\choose 2}^2$ Steiner vertices in addition to the terminals, then distortion $1$ can be achieved. They further showed that for graphs with constant treewidth, $O(k^2)$ Steiner points will suffice for distortion $1$.
Cheung, Gramoz and Henzinger \cite{CGH16} showed that allowing $O(k^{2+\frac2t})$ Steiner vertices, one can achieve distortion $2t-1$ (in particular distortion $O(\log k)$ with $O(k^2)$ Steiners). For planar graphs, Cheung et. al. achieved $1+\eps$ distortion with $\tilde{O}((\frac k\epsilon)^2)$ Steiner points.

There is a long line of work focusing on preserving the cut/flow structure among the terminals, by a graph minor. See 
\cite{Moitra09,LM10,CLLM10,MM10,EGKRTT14,Chuzhoy12,KR13,AGK14,GHP17,KR17}.

Finally, there were works studying metric embeddings and metric data structures concerning with preserving distances among terminals, or from terminals to other vertices, out of the context of minors. See \cite{CE05,RTZ05,GNR10,KV13,EFN15S,EFN15T,BFN16}.

\subsection{Technical Ideas}
We use the ball growing algorithm presented in \cite{KKN15} (also used by \cite{Cheung17}), with adjusted parameters. The algorithm work in rounds. In each round, by turn, each terminal $t_j$ increases the radius $R_j$ of its ball-cluster $V_j$ in attempt to add more vertices to its cluster $V_j$. Once a vertex joins some cluster, it will remain there. 
In round $\ell$, the radii are (independently) sampled according to exponential distribution with mean $D\cdot r^\ell$, where  $r=1+\frac{O(1)}{\ln k}$ and $D=r-1$. In each consecutive round, the mean of the distribution is multiplied by $r$. This extremely slow growth rate  allows us to control (w.h.p) the round in which each vertex will be covered (that is, join some cluster).
Specifically, for vertex $v$ whose closest terminal is at distance $D(v)$, w.h.p. $v$ is covered somewhere between round $\ensuremath{\log_{r}\left(\Omega(D(v))\right)}$ to round
$\ensuremath{\log_{r}\left(O(D(v)\right)}$. In particular, $v$ will be covered by terminal $t$ at distance at most $O(D(v))$ from $v$. Furthermore, every vertex $v'$ that is covered simultaneously with $v$ will be also at distance at most $O(D(v))$ from $t$.

In the end of the algorithm, when all the vertices are covered, we contract each cluster into a single vertex to get a minor graph $M$ on the terminals. The weight in $M$ of the edge $\{t_i,t_j\}$ (if exist) is simply set to $d_G(t_i,t_j)$.
In order to bound the distance in the minor graph between two terminals $t,t'$, we partition the shortest path $P_{t,t'}$ from $t$ to $t'$ into a set of intervals $\mathcal{Q}$. The length $|Q|$ of each interval $Q\in\mathcal{Q}$ will be $\Theta(\frac{D(Q)}{\ln k})$, where $D(Q)$ is the distance from $Q$ to its closest terminal. In particular, $Q$ will have the property that if some vertex $v\in Q$ is covered by $t$ at round $\ell$, then with probability at least $0.8$, all of $Q$ is covered by $t$ (at round $\ell$).

We can show that the expected number of terminals covering the vertices of $Q$ is constant.
In fact, Cheung \cite{Cheung17} argued that w.h.p every interval $Q$ is covered by at most $O(\log k)$ different terminals. This is the reason he pays additional $\log k$ factor on the distortion. We will use a subtler argument in order to spare this $\log k$ factor.

We will analyze the covering of all the intervals simultaneously.  
Consider round $\ell$, where terminal $t_j$ grows its cluster. Note that $t_j$ might cover vertices from different intervals. Let $Q_j^\ell$ be the interval containing the closest vertex to $t_j$, among the vertices of $P_{t,t'}$ that were covered by $t_j$ at round $\ell$.
The vertices covered by $t_j$ at round $\ell$ will create a detour $\mathcal{D}_{j}^{\ell}$, which will be charged upon $Q_{j}^{\ell}$. 
The sum of the lengths of all the detours created during the algorithm can be used to bound $d_M(t,t')$. The length of each $\mathcal{D}_{j}^{\ell}$ equals $O(\log k)\cdot Q_{j}^{\ell}$.

In each step at most one interval will be charged. All the covered vertices not in $Q_{j}^{\ell}$ will be covered free of charge.
We define a cost function $f$ which is defined by a linear combinations of all the charges upon all the intervals. Essentially $f$ is proportional to the length of all the created detours, and thus can be used to bound $d_M(t,t')$.
The next step is to use a concentration bounds to show that while some intervals might be charged for large number of detours, on average the cost function will not exceed the expectation by much. However, as the charges upon different intervals are strongly dependent, this requires a subtle argument.

\section{Preliminaries}
\appendixref{appendix:key} contains a summary of all the definitions and notations we use. The reader is encouraged to refer to this index  while reading. 

We consider undirected graphs $G=(V,E)$ with positive edge weights
$w: E \to \R_{\geq 0}$. Let $d_{G}$ denote the shortest path metric in
$G$. Let $B_{G}(v,r)=\{u\in V\mid d_{G}(v,u)\le r\}$ be the ball around $v$ in $G$ with radius $r$.
For a subset of vertices $A\subseteq V$, let $G[A]$ denote the \emph{induced graph} on $A$. 
Fix $K=\{t_1,\dots,t_k\}\subseteq V$ to be a set of \emph{terminals}. For a vertex $v$, $D(v)=\min_{t\in K}d_{G}(v,t)$ is the distance from $v$ to the closest terminal. 

\begin{figure}[]
	\centering{\includegraphics[scale=0.85]{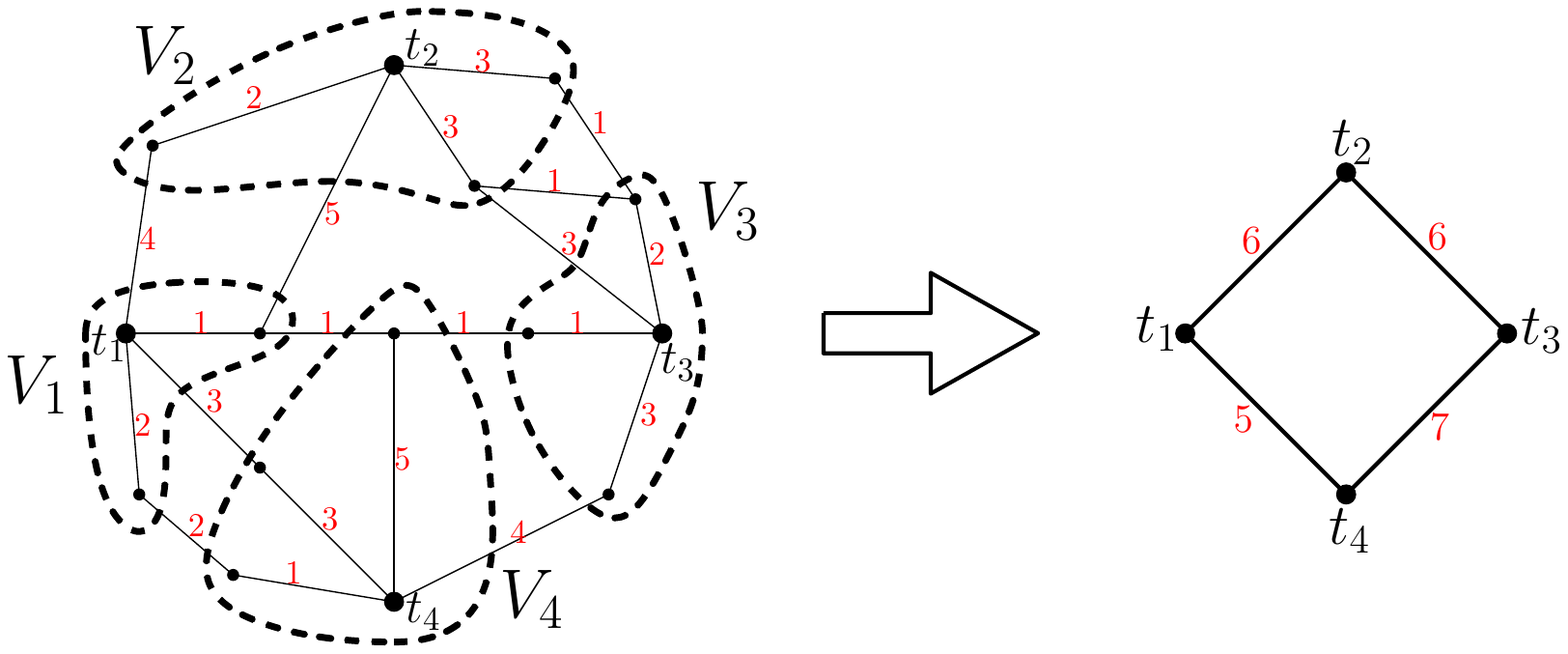}} 
	\caption{\label{fig:contraction}\small \it 
		The leftmost side of the figure contains a weighted graph $G=(V,E)$, with weights specified in red, and four terminals $\{t_1,t_2,t_3,t_4\}$.
		The dashed black curves represent a terminal partition of the vertex set $V$ into the subsets $V_1,V_2,V_3,V_4$.
		The right side of the figure represent the minor $M$ induced by the terminal partition. The distortion is realized between $t_1$ and $t_3$, and is $\frac{d_M(t_1,t_3)}{d_G(t_1,t_3)}=\frac{12}4=3$.
	}
\end{figure}

A graph $H$ is a \emph{minor} of a graph $G$ if we can obtain $H$ from
$G$ by edge deletions/contractions, and vertex deletions. A partition $\{V_1,\dots,V_k\}$ of $V$ is called a \emph{terminal partition} (w.r.t $K$) if for every $1\le i\le k$, $t_i\in V_i$, and the induced graph $G[V_i]$ is connected. See \figureref{fig:contraction} for illustration.
The \emph{induced minor} by terminal partition  $\{V_1,\dots,V_k\}$, is a minor $M$, where each set $V_i$ is contracted into a single vertex called (abusing  notation)  $t_i$. 
Note that there is an edge in $M$ from  $t_i$ to $t_j$ iff there are vertices $v_i\in V_i$ and $v_j\in V_j$ such that $\{v_i,v_j\}\in E$. 
We determine the weight of the edge $\{t_i,t_j\}\in E(M)$ to be $d_G(t_i,t_j)$. Note that by the triangle inequality, for every pair of (not necessarily neighboring) terminals $t_i,t_j$, it holds that $d_M(t_i,t_j)\ge d_G(t_i,t_j)$.
The \emph{distortion} of the induced minor is  $\max_{i,j}\frac{d_M(t_i,t_j)}{d_G(t_i,t_j)}$.

\subsection{Exponential Distribution}

$\Exp(\lambda)$ denotes the \emph{exponential distribution} with mean
$\lambda$ and density function $f(x)=\frac{1}{\lambda}e^{-\frac{x}{\lambda}}$ for $x\ge0$. $X\sim\Exp(\lambda)$ denotes the the random variable $X$ distributed according to $\Exp(\lambda)$.
By $a+c\cdot \Exp(\lambda)$ we denote a distribution where we sample $X\sim\Exp(\lambda)$, and return $a+c\cdot X$.
A useful property of exponential distribution is \emph{memoryless}: let $X\sim\Exp(\lambda)$, for every $a,b\ge 0$, $\Pr[X\ge a+b\mid X\ge a]=\Pr[X\ge b]$. In other words, given that $X\ge a$, it holds that $X\sim a+\Exp(\lambda)$. 
Another useful property of exponential distribution is \emph{closeness under scaling}, that is $c\cdot\Exp(\lambda)$ is equal to $\Exp(c\lambda)$.
We will use the following concentration bounds, the proof of which appears in \appendixref{app:ConProof}
\begin{restatable}{lemma}{LemmaExpCon}
	\label{lem:ExpConcentration}
	Suppose $X_{1},\dots,X_{n}$'s are independent random
	variables, where each $X_{i}$ is distributed according to $\Exp(\lambda_{i})$.
	Let $X=\sum_{i}X_{i}$ and $\lambda_{M}=\max_i\lambda_{i}$. Set $\mu=\mathbb{E}\left[X\right]=\sum_{i}\lambda_{i}$.
	\begin{align*}
	\text{For }a\ge2\mu\text{, }~~~~~~ & \Pr\left[X\ge a\right]\le\exp\left(-\frac{1}{2\lambda_{M}}\left(a-2\mu\right)\right)~.\\
	\text{For }a\le\frac{\mu}{2}\text{, }~~~~~~ & \Pr\left[X\le a\right]\le\exp\left(-\frac{1}{\lambda_{M}}\left(\frac{\mu}{2}-a\right)\right)~.
	\end{align*}
\end{restatable}

\section{Algorithm}
We will assume that $\min_{v}D(v)=1$, as we can scale all the weights by a constant (and rescale appropriately the output). In addition, we will assume that the number of terminals $k$ is larger than a big enough constant, as otherwise the algorithm of \cite{KKN15} is asymptotically optimal.

Before executing our algorithm, we will make some preprocessing to the input graph $G$.
Our first step will be to use the algorithm of Krauthgamer, Nguyen and Zondiner \cite{KNZ14}
to obtain a minor of the input graph such that all terminal distances are preserved exactly, while the minor contains at most $2\cdot {k\choose 2}^2<\frac{k^4}2$ steiner vertices.
Let $P_{t,t'}$ be an arbitrary shortest path from $t$ to $t'$. Our next preprocessing step will be to ensure that every edge $e$ on  $P_{t,t'}$ has weight at most $\frac{c_{w}}{\ln k}\cdot d_{G}(t,t')$, where $c_w=\frac{1}{2400}$.
This can be achieved by subdividing larger edges by adding additional
vertices of degree two in the middle of large edges. 
This modification will require at most $\frac{\ln k}{c_{w}}$ vertices
per path, an thus a total of $\frac{k^2\ln k}{c_{w}}<\frac{k^{4}}{3}$ additional vertices. Thus, after this modification, the graph will contain at most $k^4$ vertices.
As we added only Steiner vertices of degree $2$, every induced minor by terminal partition of the new graph, will be a minor of the original graph as well.
From now on, we will abuse notation and let $G$ be the resulting graph (after both modifications) as if it were the original one.

After we finish with the preprocessing, we are ready to execute \algref{alg:mainSPR}, which is the same as the algorithm used by \cite{KKN15} (and \cite{Cheung17}), with adjusted parameters.
Each terminal $t_j$, will be associated with a radius $R_j$ and cluster $V_j\subset V$. 
During the algorithm we will iteratively grow clusters $V_1,\dots,V_k$ around the terminals. Once some vertex $v$ joins some cluster $V_j$, it will stay there. When all the vertices are clustered, the algorithm terminates. 
Initially the cluster $V_j$ contains only the terminal $t_j$, while $R_j$ equals $0$.
The algorithm will have rounds, where each round consist of $k$ steps. In step $j$ of round $\ell$, we sample a number $q_j^\ell$ according to distribution $\Exp(D\cdot r^\ell)$ (note that the mean of the distribution grows by a factor of $r$ in each round). The radius $R_j$ grows by  $q_j^\ell$. We consider the graph induced by the unclustered vertices $V_\perp$ union $V_j$. Every unclustered vertex of distance at most $R_j$ from $t_j$ in $G[V_\perp\cup V_j]$ joins $V_j$.

\begin{algorithm}[!ht]
	\caption{$M=\texttt{Steiner-point-removal}(G=(V,E),w,K=\{t_1,\dots,t_k\})$}\label{alg:mainSPR}
	\begin{algorithmic}[1]
		\STATE Set $r\la1+\delta / \ln k$, where $\delta = \nicefrac{1}{20}$.
		\STATE Set $D\la \frac \delta{\ln k}$.
		\STATE For each $j\in [k]$, set $V_j\la\{t_j\}$, and set $R_j\la~ 0$.
		\STATE Set $V_\perp~\la~V\setminus\left(\cup_{j=1}^k V_j\right)$.
		\STATE Set $\ell ~\la~ 0$.
		
		\WHILE{$\left(\cup_{j=1}^k V_j\right) ~\neq~ V$}
			\FOR {$j$ from $1$ to $k$}
				\STATE Choose independently at random $q^\ell_j$ distributed according to $\Exp(D\cdot r^\ell)$.
				\STATE Set $R_j\la R_j+q^\ell_j$.
				\STATE Set $V_j\la B_{G[V_\perp \cup V_j]}(t_j,R_j)$.	.\hfill\emph{//~This is the same as $V_j\la V_j\cup B_{G[V_\perp \cup V_j]}(t_j,R_j)$.}					
				\STATE Set $V_\perp\la V\setminus\left(\cup_{j=1}^k V_j\right)$.
			\ENDFOR
			\STATE $\ell\la\ell+1$.
		\ENDWHILE
		\RETURN the terminal-centered minor $M$ of $G$ induced by $V_1,\ldots,V_k$.
	\end{algorithmic}	
\end{algorithm}

\begin{theorem}\label{thm:mainSPR}
	With probability $1-O\left(\frac{1}{k}\right)$, in the minor graph $M$ returned by \algref{alg:mainSPR}, it holds that for every two terminals $t,t'$, $d_{M}(t,t')\le O\left(\log k\right)\cdot d_{G}(t,t')$.
\end{theorem}

\section{Covering properties}

We say that vertex $v$ is \emph{covered} if $v\in\cup_{j}V_{j}$. If $v$ joins $V_{j}$ at round $\ell$, we say that $v$ was covered by $t_{j}$ at round $\ell$. 
In this section we upper and lower bound the round in which each vertex is covered. This will imply that every vertex $v$ is covered by a terminal $t$ at distance at most $O(D(v))$. Furthermore, we will show that if vertices $v$ and $v'$ were covered by terminal $t$ at the same round, then $d_G(v,t)$ and $d_G(v',t)$ are asymptotically equal. 

We denote by $\ECUB$ (CUB for covering upper bound)
the event that every vertex $v$ was already covered after the $\left\lfloor \log_{r}\left(4 D(v)\right)\right\rfloor $
round.
\begin{lemma}
	\label{lem:vCoveredInTime} $\Pr\left[\ECUB\right]\ge1-\frac{1}{k}$.
\end{lemma}
\begin{proof}
	Fix some vertex $v$. We will show that the probability that $v$
	remains uncovered after the $m=\left\lfloor \log_{r}\left(4 D(v)\right)\right\rfloor $
	round is bounded by $k^{-5}$. Since there are at most $k^{4}$
	vertices, the lemma will follow by the union bound. Let $t_{v}$ be the
	closest terminal to $v$, and denote by $P_{v}$ the shortest path from
	$t_{v}$ to $v$ in $G$ (which has length $D(v)$). 
	Denote by $u_{*}$ the currently covered vertex farthest away from $t_{v}$ on $P_{v}$, by $t_{*}$ the terminal covering $u_{*}$, and by $R_{t_*}$ the radius currently associated with $t_*$. Set
	\[
		d_{*}=d_{G}(t_{v},u_{*})+\left(R_{t_{*}}-d_{G[V_{t_{*}}]}(u_{*},t_{*})\right)~.
	\]
	$d_{*}$ is the \emph{effective} covered part of $P_{v}$. Note that there might be no vertex at distance exactly $R_{t_*}$ from $t_{*}$ to cover. However, if we could add additional vertex at distance $d_*$ from $t_v$, it would be currently covered by $t_*$.
	See \figureref{fig:dStar} for illustration.
\begin{figure}[H]
	\centering{\includegraphics[scale=0.7]{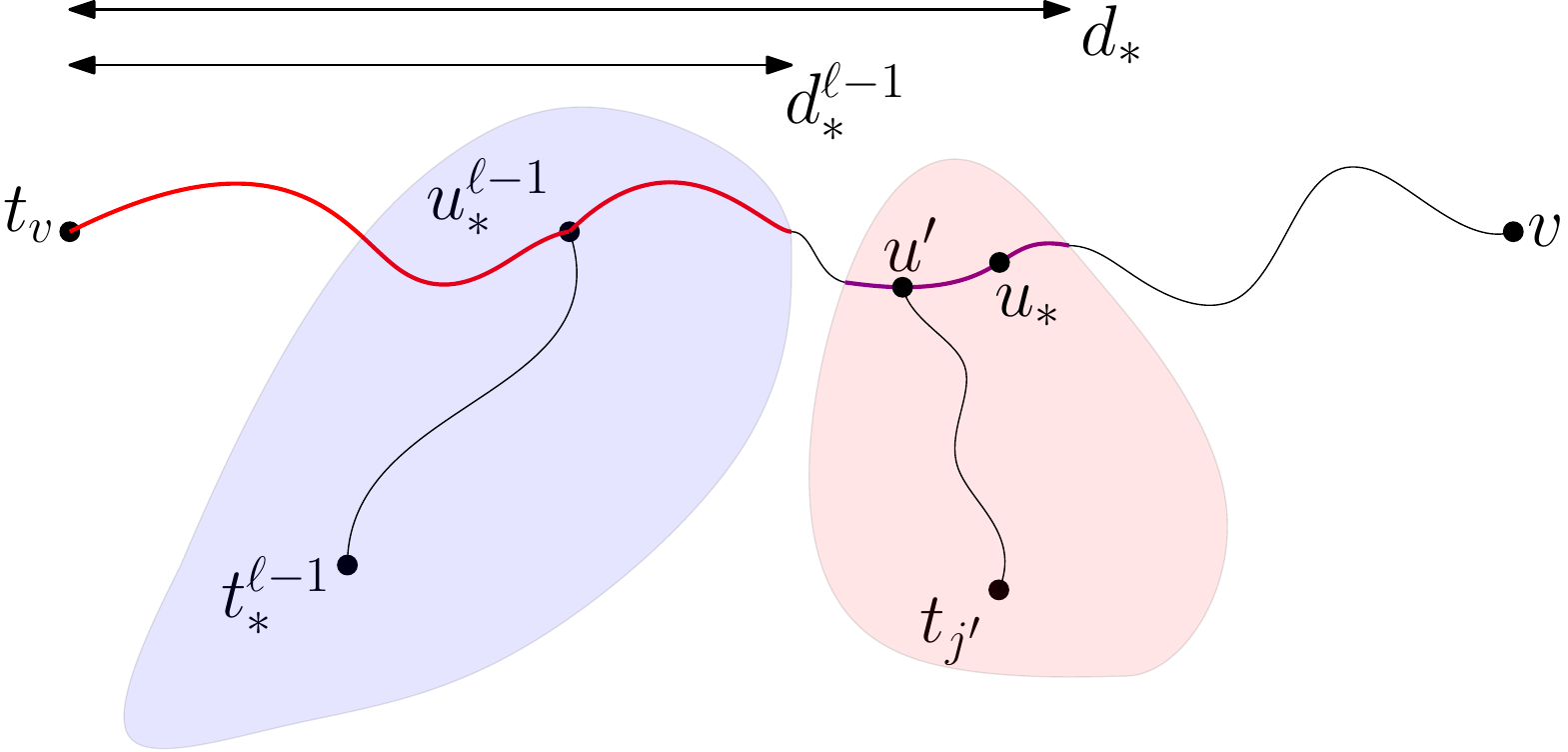}} 
	\caption{\label{fig:dStar}\small
		\it 
		At the end of the $\ell-1$, the farthest covered vertex on $P_v$ was $u_{*}^{\ell-1}$, who was covered by $t_{*}^{\ell-1}$. 
		$d_{*}^{\ell-1}$ is the length of the effective covered part of $P_v$, which is farther away from $u_{*}^{\ell-1}$ by  $R_{t_{*}^{\ell-1}}-d_{G[V_{t^{\ell-1}_*}]}(u^{\ell-1}_{*},t^{\ell-1}_{*})$ along $P_v$.
		At round $\ell$, either $t_{*}^{\ell-1}$ increase $R_{t_{*}^{\ell-1}}$ (and therefore $d_{*}$) according to $\Exp(D\cdot r^{\ell})$, or some new vertex $u'$ is covered by terminal $t_{j'}$ and then $d_*$ increase additionally according to distribution $\Exp(D\cdot r^{\ell})$. 
	}
\end{figure}

	Consider round $\ell$, we argue that the increase of $d_{*}$ during round $\ell$ is lower bounded by random variable distributed according to $\Exp(D\cdot r^{\ell})$.
	Let $u_{*}^{\ell-1}$ be $u_{*}$ by the end of the $\ell-1$ round, $t_{*}^{\ell-1}$ be the terminal covering $u_{*}^{\ell-1}$, and $d^{\ell-1}_{*}$ be the value of $d_{*}$ by the end of the $\ell-1$ round.
	Let $j$ be such that $t_j=t_{*}^{\ell-1}$. If by the $j$'s step of the $\ell$'s round, $u_{*}^{\ell-1}$ is still the farthest vertex covered on $P_v$ (that is $u_{*}=u_{*}^{\ell-1}$), then $d_{*}$ is growing by $q_j^\ell$ (exactly as $R_j$) which is  distributed according to $\Exp(D\cdot r^{\ell})$.
	Otherwise, let $u'$ be the first vertex on $P_v$ further than $u_{*}^{\ell-1}$ to be covered by terminal $t_{j'}$. 
	It holds that  
	\begin{align*}
	\ensuremath{d_{G[V_{\perp}\cup V_{j}]}(t_{_{*}}^{\ell-1},u')} & \le d_{G[V_{\perp}\cup V_{j}]}(t_{_{*}}^{\ell-1},u_{_{*}}^{\ell-1})+d_{G}(u_{_{*}}^{\ell-1},u')\\
	& =d_{G}(t_{v},u_{_{*}}^{\ell-1})+R_{t_{_{*}}^{\ell-1}}-d_{*}^{\ell-1}+d_{G}(u_{_{*}}^{\ell-1},u')=d_{G}(t_{v},u')+R_{t_{_{*}}^{\ell-1}}-d_{*}^{\ell-1}~.
	\end{align*}
	Therefore $d_{G}(t_{v},u')>d_{*}^{\ell-1}$, as otherwise $\ensuremath{d_{G[V_{\perp}\cup V_{j}]}(t_{_{*}}^{\ell-1},u')}\le  R_{t_{_{*}}^{\ell-1}}$, contradiction to the fact that $u'$ was not already covered by $t_{_{*}}^{\ell-1}$.
	By the memoryless property of exponential distribution, given that $t_{j'}$ covered $u'$, $R_{j'}$ and therefore $d_*$ will increase additively according to distribution $\Exp(D\cdot r^{\ell})$. 
	Note that $d_*$ never decreases. We conclude that until $d_*$ reaches $D(v)$, the growth
	of $d_{*}$ in round $\ell$ is lower bounded by a random variable distributed according to $\Exp(D\cdot r^{\ell})$.
	
	Let $X_{0},X_{1},\dots,X_{m}$ be independent random variables, where $X_{\ell}\sim\Exp(D\cdot r^{\ell})$, and $X=\sum_{\ell=0}^{m}X_\ell$. The probability that $v$ is not covered after $m$ rounds is lower bounded by the probability that $X<D(v)$. The mean of $X$ is
	\[
	\mu=\mathbb{E}\left[X\right]=\sum_{\ell=0}^{m}\mathbb{E}\left[X_{\ell}\right]=D\cdot\sum_{\ell=0}^{m}r^{\ell}=D\cdot\frac{r^{m+1}-1}{r-1}\ge r^{m+1}-1\ge\frac{r^{m+1}}{2}~.
	\]
	The maximal mean of $X_i$ is $\lambda_M=D\cdot r^m$. Note also that $r^{m+1}>r^{\log_{r}\left(4D(v)\right)}=4D(v)$, thus $D(v)<\frac{r^{m+1}}{4}$.
	By \lemmaref{lem:ExpConcentration} we conclude
	\begin{align*}
	\Pr\left[X\le D(v)\right] & \le\exp\left(-\frac{1}{\lambda_{M}}\left(\frac{\mu}{2}-D(v)\right)\right)\\
	& \le\exp\left(-\frac{1}{ D\cdot r^{m}}\left(\frac{r^{m+1}}{2}-\frac{r^{m+1}}{4}\right)\right)<\exp\left(-\frac{\ln k}{4\cdot\delta}\right)=k^{-5}~.
	\end{align*}
\end{proof}
	
	Set $\cCE=\frac13$ (CE for covered early). 
	We denote by $\ECE$ the event that for some vertex $v$ and terminal $t$, $t$ covered $v$ before the $\left\lfloor \log_{r}(\cCE\cdot d_{G}(v,t))\right\rfloor$ round.
\begin{lemma}
	\label{lem:coveredEarly}$\Pr\left[\ECE\right]\le k^{-3}$.
\end{lemma}
\begin{proof}
	We denote by $\ECE_{v,t}$ the event that the vertex $v$  was covered by the terminal $t$ before the $\left\lfloor \log_{r}(\cCE\cdot d_{G}(v,t))\right\rfloor$ round. Note that $\ECE=\cup_{v,t}\ECE_{v,t}$. We will show that $\Pr\left[\ECE_{v,t}\right]\le k^{-8}$, and the lemma will follow by union bound.
	
	Fix some vertex $v$ and terminal $t$. Denote by $R_{t}^{\ell}$
	the value of $R_{t}$ after the $\ell$'th round. $\ECE_{v,t}$ might  occur only if $R_{t}^{m}$ is at least
	$d_{G}(t,v)$ for $m=\left\lfloor \log_{r}(\cCE\cdot d(t,v))\right\rfloor $.
	The growth of $R_{t}$ at round $\ell$ is according to  $\Exp(D\cdot r^{\ell})$, where all the rounds are independent. Hence $R_{t}^{m}\sim\sum_{\ell=0}^{m}\Exp(D\cdot r^{\ell})$. 
	It holds that $\mathbb{E}\left[R_{t}^{m}\right]=\sum_{\ell=0}^{m}D\cdot r^{\ell}=D\cdot\frac{r^{m+1}-1}{r-1}\le r^{m+1}$.
	By \lemmaref{lem:ExpConcentration}, we conclude
	\begin{align*}
	\Pr\left[\mathcal{E}_{v,t}^{\text{CE}}\right]\le\Pr\left[R_{t}^{m}\ge d_{G}(t,v)\right] & \le\exp\left(-\frac{1}{2\cdot D\cdot r^{m}}\left(d(t,v)-2\cdot\mathbb{E}\left[R_{t}^{m}\right]\right)\right)\\
	& \le\exp\left(-\frac{1}{2\cdot D\cdot r^{m}}\left(\frac{1}{\cCE}\cdot r^{m}-2\cdot r^{m+1}\right)\right)\\
	& =\exp\left(-\frac{\ln k}{2\delta}\left(\frac{1}{\cCE}-2r\right)\right)<\exp\left(-8\ln k\right)=k^{-8}~.
	\end{align*}
\end{proof}
\begin{corollary}
	\label{cor:sameRoundSameDv}Assuming $\ECUB$ and $\overline{\ECE}$,
	for every two vertices $v,v'$ who both were covered by terminal $t$
	at round $\ell$, it holds that 
	 $d_{G}(t,v')=O(D(v))$.
\end{corollary}
\begin{proof}
	As we assumed $\ECUB$, $v$ necessarily was covered
	until round $\left\lfloor \log_{r}\left(4 D(v)\right)\right\rfloor $,
	that is $\ell\le\left\lfloor \log_{r}\left(4 D(v)\right)\right\rfloor $.
	From the other hand, $\overline{\ECE}$ implies
	$\ell>\left\lfloor \log_{r}(\cCE\cdot d_{G}(v',t))\right\rfloor $.
	We conclude that $\log_{r}\left(\cCE\cdot d_{G}(v',t)\right)\le\log_{r}\left(4 D(v)\right)$,
	and therefore $d_{G}(t,v')<\frac{4}{\cCE}\cdot D(v)= 12\cdot D(v)$.
\end{proof}

\section{Clustering Analysis}
In this section we describe in detail the probabilistic process of growing clusters, and define a charging scheme that will be used to bound the distortion.

Consider two terminals
$t$ and $t'$. Let $P_{t,t'}=\left\{ t=v_{0},\dots,v_{L}=t'\right\} $
be the shortest path from $t$ to $t'$ in $G$.
We can assume that there are no terminals in $P_{t,t'}$ other than $t,t'$. This is because if we will prove that for every pair of terminals  $t,t'$ such that $P_{t,t'}\cap K=\{t,t'\}$ it holds that $d_{M}(t,t')\le O(\log k)\cdot d_{G}(t,t')$, the triangle inequality will imply this property for all pairs of terminals.

Set $P=\left\{ v_{1},\dots,v_{L-1}\right\} $
to be the path $P_{t,t'}$ without its boundaries $t,t'$.
For a sub interval
$Q=\left\{ v_{a},\dots,v_{b}\right\} \subseteq P$, the \emph{internal length} is 
$L(Q)=d_{G}(v_{a},v_{b})$, and the \emph{external length} is $ L^{+}(Q)=d_{G}(v_{a-1},v_{b+1})$.
Set $\cint=\frac{\cCE}{10}=\frac{1}{30}$ (``int'' for interval).
We partition the vertices in $P$ into sub intervals $\mathcal{Q}$, with the property that
each $Q\in\mathcal{Q}$ will contain a vertex $u_{Q}\in Q$ such that
$L(Q)\le\frac{\cint\delta}{\ln k}D(u_{Q})\le L^{+}(Q)$:
Such a partition could be constructed as follows. Sweep
along the interval $P$ in a greedy manner, after partitioning the prefix $v_{1},\dots,v_{h-1}$,
to construct the next $Q$, we set $u_{Q}=v_{h}$ and simply pick
the minimal index $s$ such that $ L^{+}(\left\{ v_{h},\dots,v_{h+s}\right\} )\ge\frac{\cint\delta}{\ln k}D(v_{h})$.
By the minimality of $s$, $L(\left\{ v_{h},\dots,v_{h+s}\right\} )\le L^{+}(\left\{ v_{h},\dots,v_{h+s-1}\right\} )\le\frac{\cint\delta}{\ln k}D(v_{h})$ (in the case $s=0$, trivially $L(\left\{ v_{h}\right\} )=0\le\frac{\cint\delta}{\ln k}D(v_{h})$).
Note that such $s$ always could be found, as $ L^{+}(\left\{ v_{h},\dots,v_{L-1}\right\} )=d_{G}(v_{h-1},t')>d_{G}(v_{h},t')\ge D(v_{h})$.

In the beginning of \algref{alg:mainSPR}, all the vertices of $P$ are
\emph{active}. 
Consider round $\ell$ in the algorithm when terminal $t_{j}$ grows
a ball to increase $V_{j}$. Specifically, it picks $q_{j}^{\ell}$
and sets $\ensuremath{R_{j}\leftarrow R_{j}+q_{j}^{\ell}}$ and
$V_{j}\leftarrow B_{G[V_{\perp}\cup V_{j}]}(t_{j},R_{j})$.
Suppose that at least one active vertex joins $V_{j}$. 
Let $a_{j}^{\ell}\in P_{t,t'}$
(resp., $b_{j}^{\ell}$) be the active vertex joining to $V_{j}$ with minimal (resp., maximal) 
index (w.r.t $P_{t.t'}$).
All the vertices
$\left\{ a_{j}^{\ell},\dots,b_{j}^{\ell}\right\}\subset P_{t,t'}$ with indices between $a_{j}^{\ell}$ to $ b_{j}^{\ell}$
become inactive. We call this set $\left\{ a_{j}^{\ell},\dots,b_{j}^{\ell}\right\}$ a
\emph{detour} $\mathcal{D}_{j}^{\ell}$ from $a_{j}^{\ell}$ to $b_{j}^{\ell}$.
See \figureref{fig:sliceChange} for illustration.
\begin{figure}[]
	\centering{\includegraphics[scale=0.7]{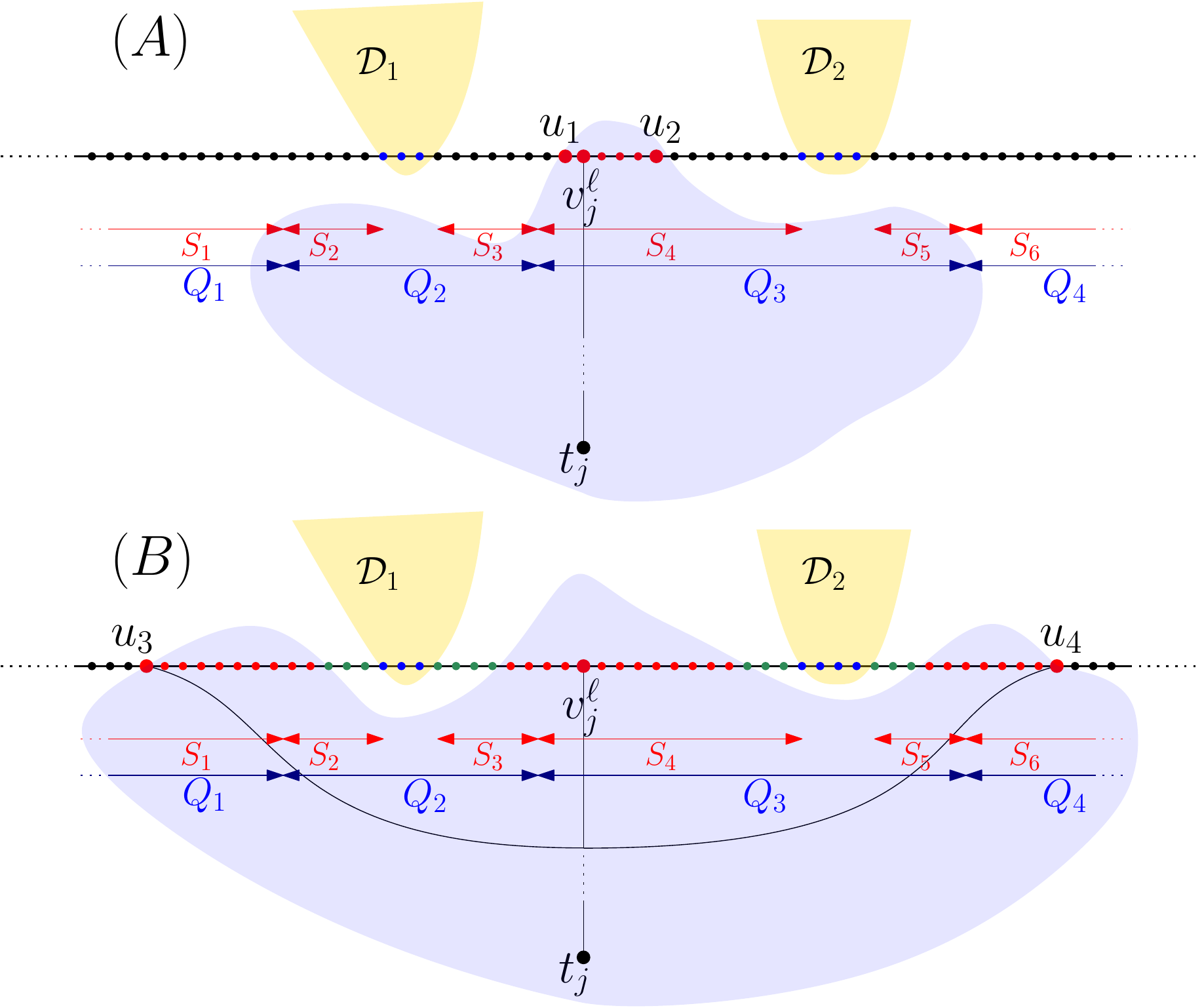}} 
	\caption{\label{fig:sliceChange}\small 
		\it 
		The figure illustrates round $\ell$ in \algref{alg:mainSPR}, when $t_j$ increases $V_j$.
		We present two scenarios for different choices of $q_j^\ell$.
		The black line is part of $P_{t,t}$ the shortest path from $t$ to $t'$. 
		The blue intervals $Q_i$ represent the subsets in $\mathcal{Q}$. The red sub-intervals $S_i$ are slices, maximal subsets of active vertices. Where $S_2,S_3\subset Q_2$ and $S_4,S_5\subset Q_3$.
		The yellow areas represent detours $\mathcal{D}_1$ and $\mathcal{D}_2$, where $Q_2$ (resp., $Q_3$) is charged for $\mathcal{D}_1$ (resp., $\mathcal{D}_2$). Note that vertices in that areas are inactive.\newline
		$t_j$ increases gradually $q_j^\ell$, the first vertex to be covered is $v_j^\ell$. In scenario (A), the growth of $q_j^\ell$ quickly terminates and sets $a_j^{\ell}=u_1$, $b_j^{\ell}=u_2$. While in scenario (B), the growth of $q_j^\ell$ continues longer, setting $a_j^{\ell}=u_3$, $b_j^{\ell}=u_4$.
		Points already inactive are colored in blue. 
		Points which currently covered by $t_j$, are colored in red.
		The green points, are points which still un-covered, but nevertheless become inactive. Points which remain active after the increase in $V_j$, are colored in black. \newline
		In scenario (A) all the vertices that become inactive, $\mathcal{D}_j^\ell$, included in $S_4$. 
		$Q_3$ is charged for $\mathcal{D}_j^\ell$. The number of slices in $Q_3$ is increased by $1$, and no other changes occur.
		In scenario (B) $\mathcal{D}_j^\ell$ contains all the vertices in $S_2,S_3,S_4,S_5$, and part of the vertices in $S_1,S_6$.
		The number of slices in $Q_2$ and $Q_3$ become $0$, while the number of slices in $Q_1$ and $Q_4$ remain unchanged.
		$Q_3$ is charged for $\mathcal{D}_j^\ell$, while its charge for $\mathcal{D}_2$ erased. Additionally, the charge of $Q_2$ for $\mathcal{D}_1$ is erased. That is, $Q_2$ will remain uncharged till the end of the algorithm ($\tilde{X}_{Q_2}=0$).
	}
\end{figure}

Within each interval $Q$, each maximal sub-interval of active
vertices is called a \emph{slice}. We denote by $\mathcal{S}(Q)$ the current
number of slices in $Q$. In the beginning of the algorithm, for every
sub interval $Q$, $\mathcal{S}(Q)=1$, while at the end of the algorithm $\mathcal{S}(Q)=0$.

For an active vertex $v$, let $q_{v}$ be the minimal choice of $q_{j}^{\ell}$,
that will force $v$ to join $V_{j}$. Let $v_{j}^{\ell}$ be the
active vertex with minimal $q_{v_{j}^{\ell}}$ (braking ties arbitrarily).
Let $Q_{j}^{\ell}\in\mathcal{Q}$ be the interval containing $v_{j}^{\ell}$. 
Similarly, let $S_{j}^{\ell}$ be the slice containing $v_{j}^{\ell}$.
We \emph{charge} $Q_{j}^{\ell}$
for the detour $\mathcal{D}_{j}^{\ell}$. We denote by $X_{Q}$ the number of
detours the interval $Q$ is currently charged for. For every detour
$\mathcal{D}_{j'}^{\ell'}$ which is contained in $\mathcal{D}_{j}^{\ell}$ (that is $a_{j}^{\ell}<a_{j'}^{\ell'}<b_{j'}^{\ell'}<b_{j}^{\ell}$ w.r.t. the order induced by $P_{t,t'}$),
we erase the detour and its charge. That is, for every $Q'\ne Q_{j}^{\ell}$,
$X_{Q'}$ might only decrease, while $X_{Q_{j}^{\ell}}$ might increase
by at most $1$ (and can also decrease as a result of deleted detours).
We denote by $\tilde{X}_{Q}$ the size of $X_{Q}$ by the end of \algref{alg:mainSPR}.
\figureref{fig:sliceChange} illustrates a single step.

Next, we analyze the change in the number of slices when $t_j$ grows its cluster at round $\ell$.
If $q_{j}^{\ell}<q_{v_{j}^{\ell}}$, then no active vertex joins $V_{j}$
and therefore $X_{Q}$ and $\mathcal{S}(Q)$ stay unchanged, for all $Q\in\mathcal{Q}$. Otherwise, $q_{j}^{\ell}\ge q_{v_{j}^{\ell}}$,
a new detour will appear, and be charged upon $Q_{j}^{\ell}$.
All the slices $S$ which are contained in $\mathcal{D}_{j}^{\ell}$ are deleted.
Every slice $S$ that intersects $\mathcal{D}_{j}^{\ell}$ but not contained
in it will be replaced by one or two new slices. If
$\mathcal{D}_{j}^{\ell}\cap S \notin\{\mathcal{D}_{j}^{\ell},S\}$,
then $S$ is replaced by a single new sub-slice $S'$. The only possibility
for a slice to be replaced by two sub-slices is if $\mathcal{D}_{j}^{\ell}\subseteq S$,
and $\mathcal{D}_{j}^{\ell}$ does not contain the extreme vertices in $S$ (see \figureref{fig:sliceChange}, scenario (A)). This can
happen only at $S_{j}^{\ell}$. We conclude that for every $Q'\ne Q_{j}^{\ell}$,
$\mathcal{S}(Q')$ might only decrease, while $\mathcal{S}(Q_{j}^{\ell})$ might increase
by at most $1$.

Let $q_{S_{j}^{\ell}}$ be the minimal choice of $q_{j}^{\ell}$,
that will force all vertices in $S_{j}^{\ell}$ to become inactive.
If $q_{j}^{\ell}\ge q_{S_{j}^{\ell}}$, then $\mathcal{S}(Q_{j}^{\ell})$ will
decrease by at least one (\figureref{fig:sliceChange}, scenario (B)). We call such occasion a \emph{success}. Otherwise,
if $q_{j}^{\ell}<q_{S_{j}^{\ell}}$ then $\mathcal{S}(Q_{j}^{\ell})$ might
increase by at most one. We call such occasion a \emph{failure} (\figureref{fig:sliceChange}, scenario (A)). 
\begin{claim}\label{clm:failProb}
	Assuming $\ell\ge\log_{r}\left(\cCE\cdot d_{G}(v_{j}^{\ell},t_{j})\right)$
	and $q_{j}^{\ell}\ge q_{v_{j}^{\ell}}$, the 
	failure probability is bounded by $p=0.2$.
\end{claim}
\begin{proof}
	Recall that there is a vertex $u_{Q_{j}^{\ell}}\in Q_{j}^{\ell}$
	such that $L(Q_{j}^{\ell})\le\frac{\cint\delta}{\ln k}D(u_{Q_{j}^{\ell}})$.
	In particular, by the triangle inequality 
	$D(v_{j}^{\ell})\ge D(u_{Q_{j}^{\ell}})-d_{G}\left(u_{Q_{j}^{\ell}},v_{j}^{\ell}\right)\ge\left(\frac{\ln k}{\cint\delta}-1\right)L(Q_{j}^{\ell})\ge\frac{\ln k}{2\cint\delta}\cdot L(Q_{j}^{\ell})$. It holds that 
	\begin{equation*}
	D\cdot r^{\ell}\ge D\cdot \cCE\cdot d_{G}(v_{j}^{\ell},t_{j})\ge D\cdot \cCE\cdot D(v_{j}^{\ell})\ge D\cdot \cCE\cdot\frac{\ln k}{2\cint\delta}\cdot L(Q_{j}^{\ell})=5\cdot L(Q_{j}^{\ell})~.
	\end{equation*}

	As $\mathcal{S}(Q_{j}^{\ell})\subseteq Q_{j}^{\ell}$, and all the vertices
	in $\mathcal{S}(Q_{j}^{\ell})$ are active, for every $u\in \mathcal{S}(Q_{j}^{\ell})$,
	$d_{G[V_{\perp}\cup V_{j}]}(v_{j}^{\ell},u)\le L(Q_{j}^{\ell})$ (we used here that $P_{t,t'}$ is a shortest path).
	Therefore, if $q_j^\ell\ge q_{v_j^\ell} + L(Q_{j}^{\ell})$, all the vertices in $S$ will be covered by $t_j$, and in particular become inactive. We conclude that  $q_{S_{j}^{\ell}}\le q_{v_j^\ell} + L(Q_{j}^{\ell})$ .
	Recall that $q_j^\ell$ is distributed according to $\Exp(D\cdot r^\ell)$. Using the memoryless property, we get:
	\begin{align*}
	\Pr\left[q_{j}^{\ell}<q_{S_{j}^{\ell}}\mid q_{j}^{\ell}\ge q_{v_{j}^{\ell}}\right] & =\Pr\left[q_{j}^{\ell}<q_{S_{j}^{\ell}}-q_{v_{j}^{\ell}}\right]=1-e^{-\left(q_{S_{j}^{\ell}}-q_{v_{j}^{\ell}}\right)/(D\cdot r^{\ell})}\\
	& \le\left(q_{S_{j}^{\ell}}-q_{v_{j}^{\ell}}\right)/(D\cdot r^{\ell})\le L(Q_{j}^{\ell})/(D\cdot r^{\ell})\le
	\frac{1}{5}=p~.
	\end{align*}
\end{proof}
\section{Bounding the Number of Failures}
Set $\varphi=|\mathcal{Q}|$. We define a \emph{cost function} $f:\mathbb{R}_+^{\varphi}\rightarrow\mathbb{R}_{+}$,
in the following way $f(x_{1},x_{2},\dots,x_{\varphi})=\sum_{i}x_{i}\cdot L^{+}(Q_{i})$.
Note that the cost function $f$ is linear and monotonically increasing
coordinate-wise.
In \sectionref{sec:DistBound} we show that the distance between $t$ to $t'$ in the minor graph $M$ can be bounded (roughly) by the total ``length'' of all the detours that the intervals were charged for by the end of \algref{alg:mainSPR}. Moreover, w.h.p. the ``length'' of detour $\mathcal{D}_j^\ell$ can be bounded by $L^{+}(Q_j^\ell)$. Thus  
$f\left(\tilde{X}_{Q_{1}},\dots,\tilde{X}_{Q_{\varphi}}\right)$ is an asymptotic bound on $d_M(t,t')$.
This section is devoted to proving the following lemma.
\begin{lemma}\label{lem:fbound}
	$\Pr\left[f\left(\tilde{X}_{Q_{1}},\dots,\tilde{X}_{Q_{\varphi}}\right)\ge43\cdot d_G(t,t') \right]\le 2\cdot k^{-3}$.
\end{lemma}
Using \claimref{clm:failProb}, one can show that for every $Q\in\mathcal{Q}$, $\mathbb{E}[\tilde{X}_Q]=O(1)$, and moreover, w.h.p. $\tilde{X}_Q=O(\log k)$. 
However, we use a concentration bound on all $\{\tilde{X}_{Q_1},\dots,\tilde{X}_{Q_\varphi}\}$ simultaneously in order to prove a stronger upper bound.

\subsection{Bounding by independent variables}
In our journey to bound  $f\left(\tilde{X}_{Q_{1}},\dots,\tilde{X}_{Q_{\varphi}}\right)$, the first step is replacing $\left(\tilde{X}_{Q_{1}},\dots,\tilde{X}_{Q_{\varphi}}\right)$ with independent variables. 
Consider the following process: we start with $\varphi$ \emph{boxes} $B_{Q_{1}},\dots,B_{Q_{\varphi}}$,
where the box $B_{Q}$ resembles the interval $Q\in\mathcal{Q}$. The
boxes will contain independent coins. Each coin has probability $p$
to get $0$ (failure), and $1-p$ to get $1$ (success). Coins can
be active and inactive. In the beginning, there is a single active
coin in each box $B_{Q}$. We toss the active coins in the boxes in
some arbitrary order. When tossing a coin from box $B_{Q}$, the tossed
coin becomes inactive. If we get $0$ we add two additional active coins
to the box $B_{Q}$. The process terminates when no active coins remain.
For a box $B_{Q}$, denote by $Z_{Q}$ the number of active coins,
by $Y_{Q}$ the number of inactive coins and by $\tilde{Y}_{Q}$ the
number of inactive coins at the end of the process. 
Let $\XCE$ be an indicator for the event $\ECE$  (recall that $\ECE$ is the event that some vertex $v$ was covered by some terminal $t$, before the $\left\lfloor \log_{r}(\cCE\cdot d_{G}(v,t))\right\rfloor$ round).
\begin{claim}
	\label{clm:CoinsDominate}For every $\alpha\in\mathbb{R}_{+}$,\\
	$$\Pr\left[f\left(\tilde{X}_{Q_{1}},\dots,\tilde{X}_{Q_{\varphi}}\right)\ge\alpha\right]\le\Pr\left[f\left(\tilde{Y}_{Q_{1}},\dots,\tilde{Y}_{Q_{\varphi}}\right)+\XCE\cdot f\left(k^{4},\dots,k^{4}\right)\ge\alpha\right]~.$$
\end{claim}
\begin{proof}
	We will treat $\XCE$ dynamically, such that $\XCE=0$ at the beginning of \algref{alg:mainSPR}, and becomes $1$ when some vertex $v$ is covered by terminal $t$ and round $\ell\le\left\lfloor \log_{r}(\cCE\cdot d_{G}(v,t))\right\rfloor$.
	The proof is done by coupling the two process of \algref{alg:mainSPR} and the coin tosses. We execute \algref{alg:mainSPR}, which implicitly induces slices and detour charges. Simultaneously, we will use \algref{alg:mainSPR} to toss coins.
	We will maintain the invariant that, as long as $\XCE=0$, $\left(Z_{Q_{1}},Y_{Q_{1}},\dots,Z_{Q_{\varphi}},Y_{Q_{\varphi}}\right)$
	is bigger than $\left(\mathcal{S}(Q_1),X_{Q_{1}},\dots,\mathcal{S}(Q_\varphi),X_{Q_{\varphi}}\right)$
	coordinate-wise. 
	In the beginning both of them are equal (to $(1,0,1,0,\dots,1,0)$).
	Consider round $\ell$, step $j$, when $t_{j}$ grows its cluster. If $q_{j}^{\ell}<q_{v_{j}^{\ell}}$
	then nothing happens, and the invariant holds. Else, $q_{j}^{\ell}\ge q_{v_{j}^{\ell}}$.
	If $\ell\le\left\lfloor \log_{r}\cCE\cdot d_{G}(v_j^\ell,t_j)\right\rfloor$, then $\XCE$ turn into $1$,	
	we unwind the coupling and continue each of the processes independently.   
	Otherwise, $\ell\ge\log_{r}\left(\cCE\cdot d_{G}(v_{j}^{\ell},t_{j})\right)$.
	We will make a coin toss from the $B_{Q_{j}^{\ell}}$ box. Let $p'$
	be the probability that $q_{j}^{\ell}<q_{S_{j}^{\ell}}$ (and thus
	$\mathcal{S}(Q_{j}^{\ell})$ might grow), recall that $p'\le p$ (\claimref{clm:failProb}). If indeed $q_{j}^{\ell}<q_{S_{j}^{\ell}}$,
	then the coin set to be $0$. Otherwise, if $q_{j}^{\ell}\ge q_{S_{j}^{\ell}}$,
	then with probability $\frac{p-p'}{1-p'}$ the coin is set to be $0$.
	Note that the probability of $0$ is exactly $p$. If the number of
	slices $\mathcal{S}(Q_{j}^{\ell})$ is increased by $1$, then the number of active
	coins $Z_{Q_{j}^{\ell}}$ increases by $1$ as well. The number of
	detours $X_{Q_{j}^{\ell}}$ charged upon $Q_{j}^{\ell}$ might increase
	by at most $1$, while the number of inactive coins $Y_{Q_{j}^{\ell}}$
	is necessarily increases by $1$. For every $Q'\ne Q_{j}^{\ell}$, $\mathcal{S}(Q')$
	and $X_{Q'}$ might only decrease, while $Z_{B_{Q'}}$ and $Y_{B_{Q'}}$
	stay unchanged. Therefore $\left(Z_{Q_{1}},Y_{Q_{1}},\dots,Z_{Q_{\varphi}},Y_{Q_{\varphi}}\right)$
	is at least $\left(\mathcal{S}(Q_1),X_{Q_{1}},\dots,\mathcal{S}(Q_\varphi),X_{Q_{\varphi}}\right)$
	coordinate-wise after the changes made at round $\ell$ step $j$ as well.
	
	At the end of the algorithm (when no slices are left), we might still
	have some active coins. In this case we will simply toss coins until no active coins remain. Note that by
	doing so $\left(Y_{Q_{1}},\dots,Y_{Q_{\varphi}}\right)$ can only
	grow. The marginal distribution on $\left(\tilde{Y}_{Q_{1}},\dots,\tilde{Y}_{Q_{\varphi}}\right)$
	is exactly identical to the original one. 
	
	We conclude:
	in the case $\XCE=0$, at the end of the \algref{alg:mainSPR}, the two process remain coupled and hence $\left(Y_{Q_{1}},\dots,Y_{Q_{\varphi}}\right)$ greater or equal than $\left(\tilde{X}_{Q_{1}},\dots,\tilde{X}_{Q_{\varphi}}\right)$ coordinate-wise. From this point on, $\left(Y_{Q_{1}},\dots,Y_{Q_{\varphi}}\right)$ can only grow. The claim follows as $f$ is monotone. 
	In the case where $\XCE=1$, the claim follows as $\tilde{X}_{Q}$ is smaller than $k^4$ for every $Q$ (as the number of vertices and therefore detours is bounded by $k^4$).
\end{proof}

\subsection{Replacing Coins with Exponential Random Variables}
Our next step is to replace each $Y_Q$ with exponential random variable. This is done in order to use concentration bounds.
Consider some box $B_{Q}$. Equivalent way to describe the probabilistic
process in $B_{Q}$ is the following. Take a single coin with failure
probability $p$, toss this coin until the number of successes exceeds
the number of failures. The total number of tosses is exactly $\tilde{Y}_{Q}$.
Note that $\tilde{Y}_{Q}$ is necessarily odd. Next we bound the probability
that $\tilde{Y}_{B}\ge2m+1$, for $m\ge1$. This is obviously upper
bounded by the probability that in a series of $2m$ tosses we had
at least $m$ failures (as otherwise the process would have stopped earlier,
in fact this true even for $2m-1$ tosses). Let $Z_{i}$ be an indicator
for a failure in the $i$'th toss. $Z=\sum_{i=1}^{2m}Z_{i}$. Note that
$\mathbb{E}\left[Z\right]=2m\cdot p$.
A bound on $Z$ follows by Chernoff inequality.
\begin{fact}[Chernoff inequality]\label{fact:Chernoff}
	Let $X_{1},\dots,X_{n}$ be i.i.d indicator variables each with probability
	$p$. Set $X=\sum_i X_{i}$ and $\mu=\mathbb{E}[X]=np$. Then for every
	$\delta\le2e-1$, $\Pr\left[X\ge(1+\delta)\right]\le\exp(-\mu\delta^{2}/4)$.
\end{fact}
\begin{align*}
\Pr\left[\tilde{Y}_{B}\ge2m+1\right]\le\Pr\left[Z\ge m\right] & =\Pr\left[Z\ge\left(1+(\frac{1}{2p}-1)\right)\mathbb{E}[Z]\right]\\
& \le\exp\left(-2m\cdot p\cdot(\frac{1}{2p}-1)^{2}/4\right)=\exp\left(-\frac{9}{40}m\right)\le\exp\left(-\frac{1}{5}m\right)~.
\end{align*}
\sloppy We conclude that the distribution of $\tilde{Y}_{B}$ is dominated
by $1+\Exp\left(10\right)$ (as for $W\sim\Exp(10)$, $\Pr\left[1+W\ge2m+1\right]=\exp\left(-\frac m5\right)$). Let $W_1,W_2,\dots,W_\varphi$ be i.i.d. variables distributed according to $\Exp(10)$,
 since all the boxes are independent and $f$ is linear and monotone
coordinate-wise, we conclude:
\begin{claim}
	\label{clm:ExpDominate}For every 	$\alpha\in\mathbb{R}_{+}$, $\Pr\left[f\left(\tilde{Y}_{B_{1}},\dots,\tilde{Y}_{B_{\varphi}}\right)\ge\alpha\right]\le\Pr\left[f\left(1,\dots,1\right)+f\left(W_{1},\dots,W_{\varphi}\right)\ge\alpha\right]$.
\end{claim}
\begin{proof}
	Set $f_{\setminus\{s\}}(x_{1},\dots,x_{s-1},x_{s+1},\dots,x_{\varphi})=\sum_{i\in[\varphi]\setminus\left\{ s\right\} }x_{i}\cdot L^{+}(Q_{i})$.
	When integrating over the appropriate measure space, it holds that
	\begin{align*}
	\Pr\left[f\left(\tilde{Y}_{1},\dots,\tilde{Y}_{\varphi}\right)\ge\alpha\right] & =\int_{\beta}\Pr\left[f_{\setminus\{1\}}\left(\tilde{Y}_{2},\dots,\tilde{Y}_{\varphi}\right)=\beta\right]\cdot\Pr\left[\tilde{Y}_{1}\cdot L^{+}(Q_{1})\ge\alpha-\beta\right]d\beta\\
	& \le\int_{\beta}\Pr\left[f_{\setminus\{1\}}\left(\tilde{Y}_{2},\dots,\tilde{Y}_{\varphi}\right)=\beta\right]\cdot\Pr\left[1+W_{1}\ge\frac{\alpha-\beta}{L^{+}(Q_{1})}\right]d\beta\\
	& =\Pr\left[f\left(1+W_{1},\tilde{Y}_{2},\dots,\tilde{Y}_{\varphi}\right)\ge\alpha\right]\\
	& \le\Pr\left[f\left(1+W_{1},1+W_{2},\tilde{Y}_{3},\dots,\tilde{Y}_{\varphi}\right)\ge\alpha\right]\\
	& \le\cdots\le\Pr\left[f\left(1+W_{1},\dots,1+W_{\varphi}\right)\ge\alpha\right]\\
	& =\Pr\left[f\left(1,\dots,1\right)+f\left(W_{1},\dots,W_{\varphi}\right)\ge\alpha\right]~.
	\end{align*}
\end{proof}

\subsection{Concentration}

Set $\Delta=d_{G}(t,t')$. It holds that 
\begin{equation*}
\Delta\le\sum_{Q\in\mathcal{Q}}L^{+}(Q)\le2\Delta~,
\end{equation*}
as every edge in $P_{t,t'}$ is counted at least once, and at most
twice in this sum. In particular $f\left(1,\dots,1\right)\le2\Delta$.
Recall that every edge in $P_{t,t'}$ is of weight at most $\frac{c_{w}}{\ln k}\cdot d_{G}(t,t')$.
In particular, for every $Q\in\mathcal{Q}$, $L^{+}(\mathcal{Q})\le L(\mathcal{Q})+\frac{2c_{w}}{\ln k}\cdot\Delta$.
For every vertex $v$ on $P_{t,t'}$, it holds that $D(v)\le\min\left\{ d_{G}(v,t),d_{G}(v,t')\right\} \le\frac{\Delta}{2}$.
Therefore for every $Q\in\mathcal{Q}$, \[
L^{+}(\mathcal{Q})\le L(\mathcal{Q})+\frac{2c_{w}}{\ln k}\cdot\Delta\le\frac{\cint\delta}{\ln k}\cdot D(u_{Q})+\frac{2c_{w}}{\ln k}\cdot\Delta\le\frac{\cint\delta}{\ln k}\cdot\frac{\Delta}{2}+\frac{2c_{w}}{\ln k}\cdot\Delta=\frac{\cint\delta}{\ln k}\cdot\Delta~.
\]

Let $\tilde{W}_{Q}\sim L^{+}(Q)\cdot\Exp\left(10\right)$.
In particular, $\tilde{W}_{Q}\sim\Exp\left(10\cdot L^{+}(Q)\right)$.
Set $\tilde{W}=\sum_{Q\in\mathcal{Q}}\tilde{W}_{Q}$. Then $f\left(W_1,\dots,W_\varphi\right)$
is distributed exactly as $\tilde{W}$. The maximal mean among the $\tilde{W}_{Q}$'s
is $\lambda_{M}=\max_{Q\in\mathcal{Q}}10\cdot L^{+}(Q)\le10\cdot\frac{\cint\delta}{\ln k}\cdot\Delta$.
The mean of $\tilde{W}$ is $\mu=\sum_{Q\in\mathcal{Q}}10\cdot L^{+}(Q)\le20\Delta$.
Set $c_{\text{con}}=3\delta \cint=\frac{1}{200}$ (con for concentration). Using \lemmaref{lem:ExpConcentration},
we conclude 
\begin{align*}
\Pr\left[f\left(W_{1},\dots,W_{\varphi}\right)\ge(c_{\text{con}}+2)20\Delta\right] & =\Pr\left[\tilde{W}\ge(c_{\text{con}}+2)20\Delta\right]\\
& \le\exp\left(-\frac{1}{2\lambda_{M}}\left((c_{\text{con}}+2)20\Delta-2\mu\right)\right)\\
& \le\exp\left(-\frac12\cdot\frac{\ln k}{ 10\cint\delta\Delta}\cdot20c_{\text{con}}\Delta\right)=k^{-\frac{c_{\text{con}}}{\delta\cint}}=k^{-3}~.
\end{align*}
Using \claimref{clm:CoinsDominate}, \claimref{clm:ExpDominate} and \lemmaref{lem:coveredEarly},
we conclude
\begin{align*}
\Pr\left[f\left(\tilde{X}_{Q_{1}},\dots,\tilde{X}_{Q_{\varphi}}\right)\ge(20c_{\text{con}}+42)\Delta\right] & \le\Pr\left[f\left(\tilde{Y}_{Q_{1}},\dots,\tilde{Y}_{Q_{\varphi}}\right)+\XCE\cdot f\left(k^{4},\dots,k^{4}\right)\ge(20c_{\text{con}}+42)\Delta\right]\\
& \le\Pr\left[f\left(\tilde{Y}_{Q_{1}},\dots,\tilde{Y}_{Q_{\varphi}}\right)\ge(20c_{\text{con}}+42)\Delta\right]+\Pr\left[\ECE\right]\\
& \le\Pr\left[f\left(1+W_{1},\dots,1+W_{\varphi}\right)\ge(20c_{\text{con}}+42)\Delta\right]+\Pr\left[\ECE\right]\\
& \le\Pr\left[f\left(W_{1},\dots,W_{\varphi}\right)\ge(c_{\text{con}}+2)20\Delta\right]+\Pr\left[\ECE\right]\\
& \le k^{-3}+k^{-3}=2\cdot k^{-3}~.
\end{align*}
Note that $20c_{\text{con}}\le1$, thus \lemmaref{lem:fbound} follows.

\section{Bounding the Distortion}\label{sec:DistBound}
Denote by $\EfBig$ the event that for some pair of terminals $t,t'$, 
$f\left(\tilde{X}_{Q_{1}},\dots,\tilde{X}_{Q_{\varphi}}\right)\ge43 \cdot d_G(t,t')$ .\footnote{We abuse notation here and use the same $\tilde{X}_{Q_{1}},\dots,\tilde{X}_{Q_{\varphi}}$ for all terminals.} By \lemmaref{lem:fbound} and union bound, $\Pr\left[\EfBig\right]\le {k \choose 2}\cdot  k^{-3}< k^{-1}$.
\begin{lemma}\label{lem:distortion}
	Assuming $\overline{\ECE},\ECUB$ and $\overline{\EfBig}$,for every pair of terminals $t,t'$,\\ $d_{M}(t,t')\le O(\log k)\cdot d_{G}(t,t')$.
\end{lemma}
\begin{proof}
Fix some $t,t'$. For every round $\ell$ and step $j$, the detour $\mathcal{D}_{j}^{\ell}$ was charged upon the interval $Q_{j}^{\ell}$.
In addition to the vertex $v_{j}^{\ell}$, $Q_{j}^{\ell}$ contains also a vertex $u_{Q_{j}^{\ell}}$
such that $D(u_{Q_{j}^{\ell}})\le\frac{\ln k}{\cint\delta}\cdot L^{+}(Q_{j}^{\ell})$.
By the triangle inequality, $D(v_{j}^{\ell})\le D(u_{Q_{j}^{\ell}})+L(Q_{j}^{\ell})= O\left(\ln k\right)\cdot L^{+}(Q_{j}^{\ell})$.
Using \corollaryref{cor:sameRoundSameDv}  the distances $d_{G}(a_{j}^{\ell},t_{j})$
and $d_{G}(b_{j}^{\ell},t_{j})$, between the terminal $t_j^\ell$ to the boundaries $a_{j}^{\ell},b_{j}^{\ell}$ of $\mathcal{D}_{j}^{\ell}$,
 are bounded by $O\left(D(v_{j}^{\ell})\right)=O\left(\ln k\right)\cdot L^{+}(Q_{j}^{\ell})$. 	

By the end of the \algref{alg:mainSPR}, all the vertices in $P=v_{1}\dots v_{L-1}$
are divided into consecutive detours $\mathcal{D}_{1},\dots,\mathcal{D}_{z}$. Detour
$\mathcal{D}_{i}$ was constructed at round $\ell_{i}$ by terminal $t_{j_{i}}$
and is from $a_{j_{i}}^{\ell_{i}}$ to $b_{j_{i}}^{\ell_{i}}$. In particular $\mathcal{D}_{i}$ was denoted $\mathcal{D}_{j_i}^{\ell_i}$ during the analysis.
See \figureref{fig:distortion} for illustration.
\begin{figure}[]
	\centering{\includegraphics[scale=0.8]{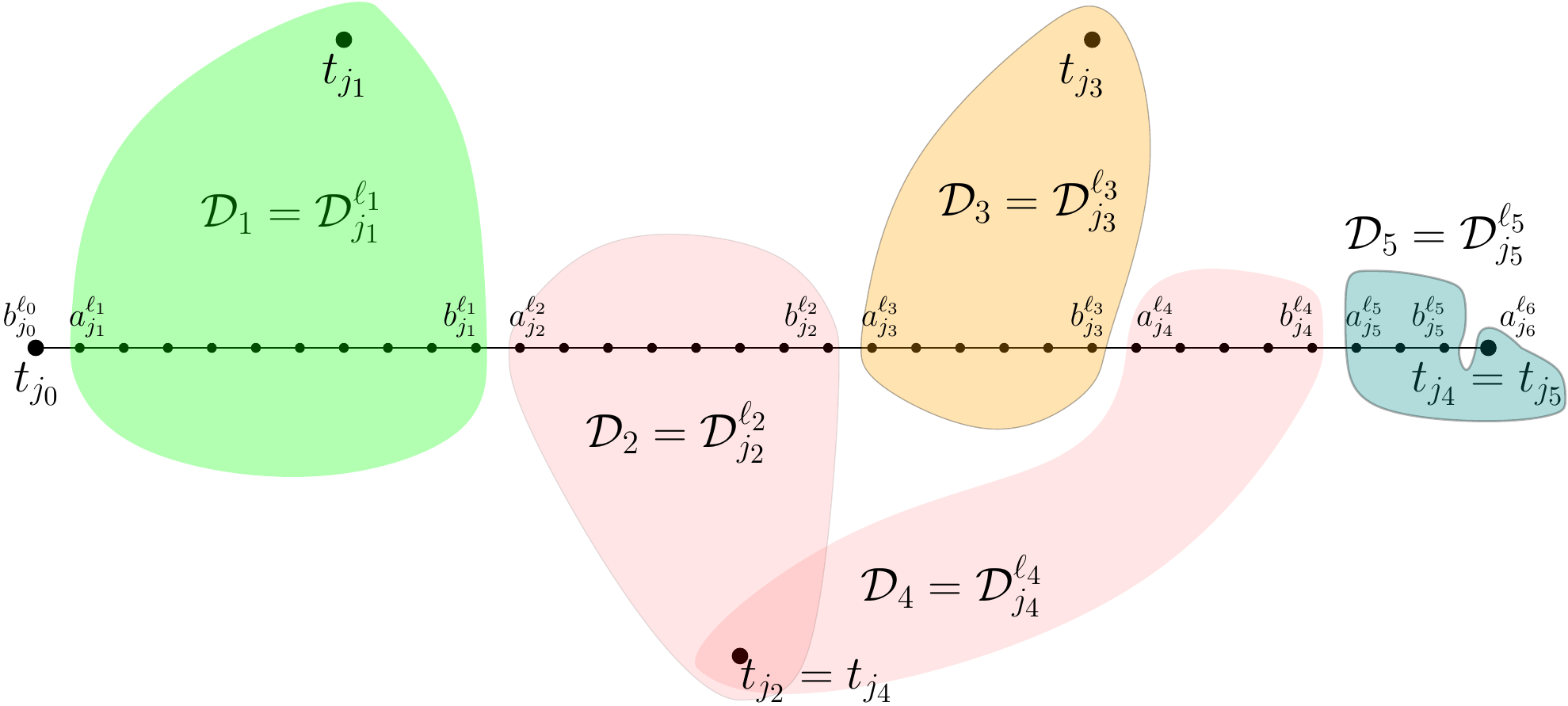}} 
	\caption{\label{fig:distortion}\small
		\it 
		The vertices $P=v_{1}\dots v_{L-1}$
		divided into consecutive detours $\mathcal{D}_{1},\dots,\mathcal{D}_{5}$. $t_{j_0},t_{j_1},t_{j_2},t_{j_3},t_{j_4},t_{j_5}$ is a walk in the terminal centered minor $M$ of $G$ (which induced by $V_1,\dots,V_k$). Note that this walk might not be a (shortest) path in $M$.
		The weight of the edge $\{t_{j_1},t_{j_2}\}$ in $M$ equals $d_G(t_{j_1},t_{j_2})$, which bounded by $d_{G}(t_{j_{1}},v_{b_{j_{1}}^{\ell_{1}}})+d_{G}(v_{b_{j_{1}}^{\ell_{1}}},v_{a_{j_{2}}^{\ell_{2}}})+d_{G}(v_{a_{j_{2}}^{\ell_{2}}},t_{j_{2}})$.
	}
\end{figure}

For every $i<z$, as $b_{j_{i}}^{\ell_{i}}\in V_{t_{j_{i}}}$,
$a_{j_{i+1}}^{\ell_{i+1}}\in V_{t_{j_{i+1}}}$ and $\left\{ b_{j_{i}}^{\ell_{i}},a_{j_{i+1}}^{\ell_{i+1}}\right\} $
is an edge in $G$, the minor graph $M$ contains an edge from $t_{j_{i}}$
to $t_{j_{i+1}}$. Set $t_{j_{0}}=t$, $t_{j_{z+1}}=t'$, $b_{j_{0}}^{\ell_{0}}=t$,
$a_{j_{z+1}}^{\ell_{z+1}}=t'$, and $ L^{+}(Q_{j_{0}}^{\ell_{0}})= L^{+}(Q_{j_{z+1}}^{\ell_{z+1}})=0$.
Note that  $\{t_{j_{0}},t_{j_{1}}\}$ and $\{t_{j_{z}},t_{j_{z+1}}\}$
are also edges in $M$.\footnote{We assume here that each terminal has an edge to itself of length
	$0$.}
We conclude
\begin{align*}
d_{M}(t,t')\le\sum_{i=0}^{z}d_{G}(t_{j_{i}},t_{j_{i+1}}) & \le\sum_{i=0}^{z}\left[d_{G}(t_{j_{i}},b_{j_{i}}^{\ell_{i}})+d_{G}(b_{j_{i}}^{\ell_{i}},a_{j_{i+1}}^{\ell_{i+1}})+d_{G}(a_{j_{i+1}}^{\ell_{i+1}},t_{j_{i+1}})\right]\\
& \le\sum_{i=0}^{z}d_{G}(b_{j_{i}}^{\ell_{i}},a_{j_{i+1}}^{\ell_{i+1}})+O\left(\ln k\right)\cdot\sum_{i=0}^{z}\left[L^{+}(Q_{j_{i}}^{\ell_{i}})+L^{+}(Q_{j_{i+1}}^{\ell_{i+1}})\right]\\
& \le\sum_{i=0}^{L-1}d_{G}(v_{i},v_{i+1})+O\left(\ln k\right)\cdot\sum_{Q\in\mathcal{Q}}2\tilde{X}_{Q}\cdot L^{+}(Q)\\
& =d_{G}(t,t')+O\left(\ln k\right)\cdot f\left(\tilde{X}_{Q_{1}},\dots,\tilde{X}_{Q_{\varphi}}\right)\\
& =O\left(\ln k\right)\cdot d_{G}(t,t')~.
\end{align*}
\end{proof}

By \lemmaref{lem:vCoveredInTime}, \lemmaref{lem:coveredEarly}, and \lemmaref{lem:fbound}, $\Pr\left[\overline{\ECE}\wedge\ECUB\wedge\overline{\EfBig}\right]\ge1-\left(\Pr\left[\ECE\right]+\Pr\left[\overline{\ECUB}\right]+\Pr\left[\EfBig\right]\right)\ge1-\frac{1}{k}-\frac{1}{k^{3}}-\frac{1}{k}>1-\frac{3}{k}$. Hence \lemmaref{lem:distortion} implies \theoremref{thm:mainSPR}.

\section{Acknowledgments}
The author would like to thank his advisors: to Ofer Neiman, for fruitful discussions, and to Robert Krauthgamer for useful comments.
\newpage
	{\small
		\bibliographystyle{alpha}
		\bibliography{SteinerBib}
	}
	
	\appendix
	
\newpage
\section{Proof of \lemmaref{lem:ExpConcentration}}\label{app:ConProof}
\LemmaExpCon*
\begin{proof}
	Set $t=\frac{1}{2\lambda_{M}}$. For each $X_{i}$, the moment generating function w.r.t $t$ equals  $\mathbb{E}\left[\exp tX_{i}\right]=\frac{1}{1-t\lambda_{i}}$.
	Using the inequality $\frac{1}{1-x}\le1+2x\le\exp(2x)$
	(for $0< x\le \frac12$) we have $\mathbb{E}\left[\exp tX_{i}\right]\le\exp\left(2t\lambda_{i}\right)$.
	Therefore,
	\begin{align*}
	\Pr\left[X>a\right] & =\Pr\left[\exp\left(tX\right)>\exp\left(ta\right)\right]\\
	& \le\mathbb{E}\left[\exp\left(tX\right)\right]/\exp(ta)\\
	& =\exp(-ta)\cdot\Pi_{i}\mathbb{E}\left[\exp tX_{i}\right]\\
	& \le\exp(-ta)\cdot\Pi_{i}\exp\left(2t\lambda_{i}\right)\\
	& =\exp(-ta+2t\mu)\\
	& =\exp\left(-\frac{1}{2\lambda_{M}}\left(a-2\mu\right)\right)~,
	\end{align*}
	where in the first inequality we use Markov inequality, and in the second
	equality we use the fact that $\left\{ X_{i}\right\} _{i}$ are independent.
	
	For the second inequality, set $t=\frac1{\lambda_M}$. It holds that $\mathbb{E}\left[\exp (-tX_{i})\right]=\frac{1}{1+t\lambda_{i}}$.
	Using the inequality $\frac{1}{1+x}\le1-\frac{x}{2}\le e^{-\frac{x}{2}}$
	(for $0< x\le 1$) we have $\mathbb{E}\left[e^{-tX_{i}}\right]\le e^{-\frac{t}{2}\lambda_{i}}$.
	Therefore,
	\begin{align*}
	\Pr\left[X<a\right] & =\Pr\left[\exp\left(-tX\right)>\exp\left(-ta\right)\right]\\
	& \le\mathbb{E}\left[\exp\left(-tX\right)\right]/\exp(-ta)\\
	& =\exp(ta)\cdot\Pi_{i}\mathbb{E}\left[\exp\left(-tX_{i}\right)\right]\\
	& \le\exp(ta)\cdot\Pi_{i}\exp\left(-\frac{t}{2}\lambda_{i}\right)\\
	& =\exp(ta-\frac{t}{2}\mu)\\
	& =\exp\left(-\frac{1}{\lambda_{M}}\left(\frac{\mu}{2}-a\right)\right)~.
	\end{align*}
\end{proof}
\newpage

\begin{multicols}{2}
\footnotesize 
\section{Index}\label{appendix:key}
\vspace{-3pt}
\subsubsection*{Preliminaries}
\begin{description}
	\item[$d_{G}$] : shortest path metric in $G$.
	\item[$B_{G}(v,r)$] : ball around $v$ in $d_G$ with radius $r$.
	\item[{$G[A]$}] : graph induced by $A$.
	\item[$K$] $=\{t_1,\dots,t_k\}$ : set of terminals.
	\item[$D(v)$] $=\min_{t\in K}d_{G}(v,t)$.
	\item[Terminal partition] : partition  $\{V_1,\dots,V_k\}$ of $V$, s.t. for every i, $t_i\in V_i$ and $V_i$ is connected.
	\item[Induced minor] : given terminal partition $\{V_1,\dots,V_k\}$, the induced minor obtained by contracting each $V_i$ into the super vertex $t_i$. The weight of the edge $\{t_i,t_j\}$ (if exist) set to be $d_G(t_i,t_j)$.
	\item[Distortion] of induced minor: $\max_{i,j}\frac{d_M(t_i,t_j)}{d_G(t_i,t_j)}$.
	\item[$\Exp(\lambda)$] : exponential distribution with mean $\lambda$.
\end{description}
\vspace{-3pt}
\subsubsection*{Assumptions}
\begin{itemize}
	\item Minimal distance between two terminals equals $1$.
	\item $k$ is larger then big enough constant.
	\item There are at most $k^{4}$ vertices in $G$.
	\item Every edge on $P_{t,t'}$ has weight at most $\frac{c_{w}}{\ln k}\cdot d_{G}(t,t')$.
	\item There are no terminals other then $t,t'$ on $P_{t,t'}$.	
\end{itemize}
\vspace{-3pt}
\subsubsection*{Constants}
\begin{description}
	\item[$\delta$] $=\frac{1}{20}$: governs the ratio $r$.
	\item[$r$] $=1+\frac{\delta}{\ln k}$: is the ratio in which the mean of the exponential  distribution grows in each round.
	\item[$D$] $=\frac{\delta}{\ln k}$: initial mean of the exponential distribution
	in round $0$.
	\item[$c_{w}$] $=\frac{\cint\delta}{4}=\frac{1}{2400}$: governs the maximum (relative) weight of
	an edge on $P_{t,t'}$.
	\item[$\cCE$]$=\frac{1}{3}$.
	The constant in $\ECE$.
	\item[$\cint$] $=\frac{\cCE}{10}=\frac{1}{30}$:
	governs the size of interval in the partition $\mathcal{Q}$ of $P$.
	\item[$\varphi$] $=|\mathcal{Q}|$: number of intervals in
	the partition $\mathcal{Q}$.
	\item[$p$] $=0.2$: upper bound on the failure probability.
\end{description}
\vspace{-3pt}
\subsubsection*{Events}
\begin{description}
	\item[$\ECUB$] : denotes that every vertex was already covered after the $\left\lfloor \log_{r}\left(4D(v)\right)\right\rfloor $
	round.
	\item[$\ECE$] : denotes that some vertex $v$ was covered by some terminal
	$t$, before the $\left\lfloor \log_{r}(\cCE\cdot d_{G}(v,t))\right\rfloor $
	round.
	\item[$\EfBig$] : denotes that for some pair of terminals $t,t'$, $\ensuremath{f\left(\tilde{X}_{Q_{1}},\dots,\tilde{X}_{Q_{\varphi}}\right)\ge43\cdot d_{G}(t,t')}$.
\end{description}

\subsubsection*{Notations}
\begin{description}
	\item[$V_{j}$] : cluster of $t_{j}$.
	\item[$R_{j}$] : radius of the cluster of $t_{j}$.
	\item[$q_{j}^{\ell}$] : growth of $R_{j}$ in the $\ell$ round.
	\item[$V_{\perp}$] : set of unclustered (uncovered) vertices.
	
	\item[$P_{t,t'}$] $=\left\{ t=v_{0},\dots,v_{L}=t'\right\} $: shortest
	path from $t$ to $t'$.
	\item[$P$] $=\left\{ v_{1},\dots,v_{L-1}\right\} $: $P_{t,t'}$ without its
	boundaries.
	\item[$L(\left\{ v_{a},v_{a+1},\dots,v_{b}\right\} )$] $=d_{G}(v_{a},v_{b})$:
	internal length.
	\item[$ L^{+}(\left\{ v_{a},v_{a+1},\dots,v_{b}\right\} )$] $=d_{G}(v_{a-1},v_{b+1})$:
	external length.
	\item[$\mathcal{Q}$] : partition of $P$ into intervals $Q$.
	\item[$u_{Q}$] : vertex in $Q\in\mathcal{Q}$ with the property $L(Q)\le\frac{\cint\delta}{\log k}D(u_{Q})\le L^{+}(Q)$.
	\item[$a_{j}^{\ell}$] : index of the leftmost active vertex covered by $t_{j}$
	at round $\ell$.
	\item[$b_{j}^{\ell}$] : index of the rightmost active vertex covered by $t_{j}$
	at round $\ell$.
	\item[$\mathcal{D}_{j}^{\ell}$] $=\left\{ v_{a_{j}^{\ell}},v_{a_{j}^{\ell}+1},\dots,v_{b_{j}^{\ell}}\right\} $:
	detour created by terminal $t_{j}$ at round $\ell$.
	\item[Slice]  maximal sub-interval (of some $Q$) of active vertices.
	\item[$q_{v}$] minimal choice of $q_{j}^{\ell}$, such that $t_{j}$
	will cover vertex $v$.
	\item[$v_{j}^{\ell}$] : vertex with the minimal $q_{v}$ (among active vertices).
	\item[$Q_{j}^{\ell}$] : interval containing $v_{j}^{\ell}$.
	\item[$S_{j}^{\ell}$] : slice containing $v_{j}^{\ell}$.
	\item[$q_{S_{j}^{\ell}}$] : minimal choice of $q_j^\ell$ that forces $t_j$ to cover all of  $S_{j}^{\ell}$.
	\item[$f(x_{1},x_{2},\dots,x_{\varphi})$] $=\sum_{i}x_{i}\cdot L^{+}(Q_{i})$:
	cost function.
	\item[$B_{Q}$] : a coin box which resembles the interval $Q$.
\end{description}
\vspace{-3pt}
\subsubsection*{Counters}
\begin{description}
	\item[$\mathcal{S}(Q)$] : (current) number of slices in interval $Q$.
	\item[$X_{Q}$] : number of detours the interval $Q$ is (currently) charged for.
	\item[$\tilde{X}_{Q}$] : number of detours the interval $Q$ is charged for
	by the end of \algref{alg:mainSPR}.
	\item[$Z_{Q}$] : number of active coins in $B_{Q}$. Each coin is active	when added to the box.
	\item[$Y_{Q}$] : number of inactive coins in $B_{Q}$. A coin become inactive after tossing.
	\item[$\tilde{Y}_{Q}$] : number of inactive coins in $B_{Q}$ by the end of the process.
\end{description}

\end{multicols}

\end{document}